\documentclass[aps,pra,reprint,amsmath,amssymb,amsfonts,eqsecnum]{revtex4-1}

\usepackage[utf8]{inputenc}

\usepackage{graphicx}
\usepackage{subcaption}

\usepackage{amsthm} 
\usepackage{mathrsfs}
\usepackage[pdfborder={0 0 0}]{hyperref}
\usepackage{color}

\theoremstyle{theorem}
\newtheorem{theorem}{Theorem}

\newtheorem{lemma}{Lemma}
\newtheorem{proposition}{Proposition}

\newcommand{\spinf}[1]{ \inf \langle #1 \rangle }

\newcommand{\Rl}{\mathbb{R}}

\newcommand{\Ss}{\mathscr{S}}

\newcommand{\hscalar}[2]{\langle #1 , \, #2 \rangle }

\newcommand{\code}[1]{{\tt #1}}

\DeclareMathOperator{\supp}{supp}

\begin{document}

\title{Quantum backflow and scattering}

\author{Henning Bostelmann}
\email{henning.bostelmann@york.ac.uk}
\affiliation{Department of Mathematics, University of York, York YO10 5DD, United Kingdom}
\author{Daniela Cadamuro}
\email{daniela.cadamuro@tum.de}
\affiliation{Zentrum Mathematik, Technische Universit\"at M\"unchen, Boltzmannstraße 3, 85748 Garching, Germany}
\author{Gandalf Lechner}
\email{LechnerG@cardiff.ac.uk}
\thanks{GL gratefully acknowledges a London Mathematical Society research in pairs grant, Ref. 41607.}
\affiliation{School of Mathematics, Cardiff University, Cardiff CF10 3AT, United Kingdom}

\date{21 July 2017}

\begin{abstract}
	Backflow is the phenomenon that the probability current of a quantum particle on the line can flow in the direction opposite to its momentum. In this article, previous investigations of backflow, pertaining to interaction-free dynamics or purely kinematical aspects, are extended to scattering situations in short-range potentials. It is shown that backflow is a universal quantum effect which exists in any such potential, and is always of bounded spatial extent in a specific sense. The effects of reflection and transmission processes on backflow are investigated, both analytically for general potentials, and numerically in various concrete examples.
\end{abstract}

\maketitle

\section{Introduction}

Backflow is the striking quantum mechanical effect that for a particle with momentum pointing to the right (with probability 1), the probability of finding the position of the particle to the right of some reference point may {\em decrease} with time. That is, probability can ``flow back'', in the direction opposite to the momentum. This effect was first described by Allcock in the context of the arrival time problem in quantum mechanics \cite{Allcock:1969}, and then discussed in detail by Bracken and Melloy \cite{BrackenMelloy:1994}. More recently, the backflow effect has attracted renewed interest \cite{EvesonFewsterVerch:2003,PenzGrublKreidlWagner:2006,Berry:2010,YearsleyHalliwellHartshornWhitby:2012,HalliwellGillmanLennonPatelRamirez:2013}, partially related to a proposed experiment to measure it \cite{PalmeroTorronteguiModugnoMuga:2013}, and partially because of its connection to other ``quantum inequalities'' appearing in quantum field theory \cite{EvesonFewsterVerch:2003,BostelmannCadamuro:oneparticle}
.

To describe backflow more precisely, consider a wave function $\psi$ defining the state of a single quantum-mechanical particle in one dimension, and its probability current density $j_\psi$. Intuitively, both the statements (with $\tilde\psi$ the Fourier transform of $\psi$)
\begin{enumerate}
     \item[a)] $\psi$ contains only positive momenta, i.e., $\supp\tilde\psi\subset\Rl_+$,
     \item[b)] $j_\psi(x)>0$ for all $x\in\Rl$,
\end{enumerate}
correspond to ``probability flowing from the left to the right''. However, a) and b) are logically independent of each other. Backflow is the observation that a) does not imply~b), that is, the current $j_\psi(x)$ can take {\em negative} values (for certain $x$), even if $\psi$ contains only {\em positive} momenta. (Note that, less surprisingly, b) does not imply a) either: Any wave function of the form $\psi(x)=e^{ipx}\varphi(x)$, where $\varphi(x)$ is real and $p>0$, satisfies b) but in general not a).) Backflow can be seen as a consequence of the uncertainty relation between position and momentum \cite{EvesonFewsterVerch:2003}. In the following, we will use the term ``right-mover'' for wave functions satisfying a).

\medskip

Of the various aspects of backflow that have been analyzed in the literature, let us recall what is known about the temporal and spatial extent of this phenomenon. In the form presented above, backflow is a purely kinematical effect, independent of a choice of dynamics or Hamiltonian. However, when a time evolution given by a self-adjoint Hamiltonian $H$ according to Schrödinger's equation $i\hbar\partial_t\psi_t=H\psi_t$ is fixed, one may study, for example, the amount of probability flowing across a reference point, say $x=0$, during a time interval $[0,T]$. Writing $j^H_\psi(t,x)=j_{\psi_t}(x)$ for the time-dependent current given by the Hamiltonian $H$, this probability is given by
\begin{align}\label{eq:BM-probability}
     {\rm p}^H_\psi(T)=\int_0^Tdt\,j^H_\psi(t,0)\,.
\end{align}
For the free Hamiltonian $H_0=\frac{1}{2m}P^2$ without potential, Bracken and Melloy found \cite{BrackenMelloy:1994} that there exists a universal dimensionless constant $0<\lambda^{H_0}<1$ such that
\begin{align}\label{eq:Bracken-Melloy-bound}
     {\rm p}^{H_0}_\psi(T) \geq -\lambda^{H_0}
\end{align}
for all normalized right-moving wave functions $\psi$ in the sense of a), and all $T>0$. The minus sign indicates that probability flows from the right to the left, i.e., this inequality is a bound on the (averaged) temporal extent of backflow. 

The backflow constant $\lambda^{H_0}$ arises as the largest positive eigenvalue of an integral operator, and has been calculated numerically to be $\lambda^{H_0}\approx0.0384517$ with growing accuracy over the years \cite{BrackenMelloy:1994,EvesonFewsterVerch:2003,PenzGrublKreidlWagner:2006}. For the construction of ``backflow states'' $\psi$ that approximate this maximal backflow, see the recent articles  \cite{YearsleyHalliwellHartshornWhitby:2012,HalliwellGillmanLennonPatelRamirez:2013}.
Backflow constants $\lambda^H$ similarly exist for interacting Hamiltonians $H$; but the kernel of the related integral operator is not explicitly known in general, and we are not aware of results on $\lambda^H$ in the interacting situation.

\smallskip

Whereas the inequality \eqref{eq:Bracken-Melloy-bound} provides a bound on the (averaged) temporal extent of backflow, one can also study its (averaged) spatial extent by considering spatial integrals of the kinematical current $j_\psi(x)$. Eveson, Fewster, and Verch \cite{EvesonFewsterVerch:2003} have shown that 
\begin{align}\label{eq:free-spatial-bound}
     \int dx\,f(x)\,j_\psi(x)\geq c_f>-\infty
\end{align}
for all normalized right-movers $\psi$ and all positive averaging functions $f(x)\geq0$. Here the function $f$ models an extended detector, generalizing the step function used in \eqref{eq:BM-probability}. Their constant $c_f$ (which has dimension of inverse time) depends on $f$, it is recalled in Eq.~\eqref{eq:bound-for-square} below.

Another general bound on (the absolute value of) the current was obtained by Muga and Leavens by expressing $j_\psi(x)$ as the expectation value of $i[\Theta(X-x),H]$ and invoking the general uncertainty relation \cite[footnote 16]{MugaLeavens:2000}.

\medskip

In this article, we extend the analysis of backflow to interacting systems, given by fairly general Hamiltonians of the form $H=\frac{1}{2m}P^2+V(X)$. We begin by recalling and refining some results on kinematical probability currents and their spatial averages in Section~\ref{sec:free}. Since the space of right-movers is no longer invariant under the time evolution if $V$ is not constant, we then propose to look at {\em asymptotic} right-movers in the sense of scattering theory, i.e., states that at very early times contain only positive momenta before scattering with the (short-range) potential. Each interaction-free right-mover $\psi$ is the incoming asymptote of an interacting state $\Omega_V\psi$, where $\Omega_V$ is the incoming M\o ller operator of the Hamiltonian $H$ with potential $V$. This familiar scattering setup is recalled in Section~\ref{sec:scattering}. That section also contains our main analytical results, which we briefly outline here.

In a scattering situation, we consider the current $j_{\Omega_V\psi}(x)$ in the interacting system that has right-moving incoming asymptote $\psi$, and study its spatial backflow, i.e., averages of the form $\int dx\,f(x)j_{\Omega_V\psi}(x)$ for positive smearing functions $f$. We show that in any short-range potential, these averages can be negative despite $\psi$ being right-moving, i.e., backflow is a {\em universal} effect which exists for any such interaction (Thm.~\ref{thm:existence-of-backflow-scattering}). This is not surprising for potentials with reflection, because reflection processes clearly produce probability flow to the left. But our result also holds for reflectionless (transparent) potentials, in which backflow exists, but turns out to be weaker than in the free case.

Generalizing \eqref{eq:free-spatial-bound}, we next study state-independent lower bounds on the averages $\int dx\, f(x)\,j_{\Omega_V\psi}(x)$ that hold for all normalized incoming right-movers $\psi$, with fixed averaging function $f(x)\geq0$. Since reflection processes amplify backflow, it is not clear a priori whether
\begin{equation}
\begin{aligned}
     \beta_V(f):=\inf\Big\{\int dx\, & f(x)\,j_{\Omega_V\psi}(x)\,: \\  & \|\psi\|=1\,,\;\psi\;\text{right-moving}\Big\}
\end{aligned}
\end{equation}
is finite. However, our analysis shows that backflow {\em is} always bounded, $\beta_V(f)>-\infty$ for all short-range potentials $V$ and all positive smearing functions $f$ (Thm.~\ref{thm:scattering-bounds}). We also derive explicit analytic estimates on the constants $\beta_V(f)$ from the asymptotic spatial behavior of the scattering solutions to Schrödinger's equation with potential $V$. Going through the analysis, it turns out that backflow can only become unbounded at large momentum. At large momentum, however, reflection is sufficiently well suppressed, which provides a heuristic understanding of this result.

Our analytic results are complemented by examples and numerical studies, presented in Section~\ref{sec:examples}. With custom computer code, supplied with this article \cite{anc-arxiv}, we study four example potentials: a delta potential as a simple extremely short range example, a rectangular potential, a reflectionless Pöschl-Teller potential, and the zero potential as a reference. Their backflow constants are calculated numerically, and their dependence on the potential strength and the position of the (Gauß type) smearing function $f$ is visualized. We also show the corresponding backflow maximizing states in that section, and discuss their properties. The numerics underlying these results is explained in more detail in the Appendix. In particular, the code can be adapted to study backflow in arbitrary short range potentials.

In Section~\ref{sec:conclusion}, we give a summary and outlook.

\section{Bounds on probability currents}\label{sec:free}

The setting of our investigation is the standard framework of quantum mechanics of one particle of mass $m>0$ in one spatial dimension. It will be convenient to work with dimensionless variables $x,p$, etc., and dimensionless functions (such as the wave function $\psi$ and the current $j_\psi$) by using a length scale $\ell$ as the unit of length, $\hbar/\ell$ as the unit of momentum, $m\ell^2/\hbar$ as the unit of time, and $\hbar^2/m\ell^2$ as the unit of energy, effectively setting $\hbar=m=1$. Thus, a square integrable wave function $\psi\in L^2(\Rl)$ defines the position probability density $|\psi(x)|^2$, and its Fourier transform $\tilde\psi(p)=(2\pi)^{-1/2} \int dx\, e^{-ipx} \psi(x)$ defines the momentum probability density $|\tilde{\psi}(p)|^2$. The operators of position, momentum, and kinetic energy are $X$, $P=-i\partial_x$, and $\frac{1}{2}P^2=-\frac{1}{2}\partial_x^2$, respectively.

With any (differentiable) wave function $\psi$, we associate its probability current density
\begin{align}\label{eq:j_psi}
     j_\psi(x)=\frac{i}{2}\left(\overline{\psi'(x)}\psi(x)-\overline{\psi(x)}\psi'(x)\right),
\end{align}
where a prime denotes a derivative with respect to $x\in\Rl$. 

\bigskip

In the context of backflow, ``right-moving'' wave functions are important, and we will write $E_\pm$ for the spectral projections of the momentum operator, corresponding to positive/negative momentum, i.e., 
\begin{align}
     \widetilde{(E_\pm\psi)}(p)=\Theta(\pm p)\tilde\psi(p)\,,
\end{align}
where $\Theta$ is the Heaviside step function. With this notation, the right-moving wave functions in statement a) in the Introduction are characterized by $E_+\psi=\psi$.

\bigskip

It is well known that locally, the backflow effect can be arbitrarily large: Given any $x\in\Rl$ and any $c>0$, there exists a normalized right moving wave function $\psi=E_+\psi$ such that $j_\psi(x)<-c$. Similarly, one can also arrange for arbitrarily large ``forward flow'', i.e., find normalized $\psi=E_+\psi$ with $j_\psi(x)>c$. 

These facts can be shown by a scaling argument; note that the (dimensionful) probability current density has the physical dimension of inverse time, so that a change of units scales its numerical value \cite{BrackenMelloy:1994}. We give here a different proof which results in more specific bounds that will be needed in the next section.

\begin{proposition}\label{proposition:jx-unbounded}{\bf (Unboundedness of $\boldsymbol{j_\psi(x)}$)}
     Let $x\in\Rl$. Then there exist sequences $\psi_n^\pm\in E_+L^2(\Rl)$ of right-moving wave functions such that
     \begin{align}
     \lim_{n\to\infty}j_{\psi_n^\pm}(x)=\pm\infty\,,
     \end{align}
     and the norms $\|\psi_n^\pm\|^2=\int dx\,|\psi^\pm_n(x)|^2$ and $\|\tilde\psi_n^\pm\|_1=\int dp\,|\tilde\psi^\pm(p)|$ are independent of $n$. 
\end{proposition}
\begin{proof}
     The unboundedness from above is a high momentum effect. To construct the sequence $\psi_n^+$, we select a right-moving wave function $\psi^+$ such that $E_+\psi^+=\psi^+$ and the current $j_{\psi^+}$ exists, and shift it to higher and higher momentum, $\tilde\psi_n^+(p):=\tilde\psi^+(p-n)$, $n\in\mathbb N$. From this construction, it is clear that $\psi_n^+=E_+\psi_n^+$, and the norms $\|\psi_n^+\|$ and $\|\tilde\psi_n^+\|_1$ are independent of $n$. Furthermore, the current of $\psi_n^+$ is
     \begin{align}\label{eq:j-shift}
	  j_{\psi_n^+}(x)=j_{\psi^+}(x)+n\,|\psi^+(x)|^2\,,
     \end{align}
     as can be quickly checked on the basis of \eqref{eq:j_psi}. When we choose $\psi^+$ such that $\psi^+(x)\neq0$ (which is clearly possible), we find $j_{\psi_n^+}(x)\to\infty$ as $n\to\infty$.
     
     To demonstrate unboundedness from below, we construct a sequence $\psi_n^-$ by superposition of a high and a low momentum state. We choose a function $\chi$ such that $\tilde\chi$ has compact support on the right half line, and $\chi(x)\neq0$. Such functions exist for any $x$, and are by construction right-movers, $E_+\chi=\chi$. We then consider the linear combinations $\tilde\psi_n^-(p):=\alpha\tilde\chi(p)+\beta\tilde\chi(p-n)$, where $n\in\mathbb N$, and $\alpha,\beta\in\mathbb C$. By construction, $E_+\psi_n^-=\psi_n^-$, and by the compact support property, we have for large enough $n$ the equalities $\|\psi_n^-\|^2=(|\alpha|^2+|\beta|^2) \| \chi\|^2$ and $\|\tilde\psi_n^-\|_1=(|\alpha|+|\beta|)\| \tilde \chi\|_1$. 
     
     It remains to choose $\alpha,\beta$ in such a way that $j_{\psi_n^-}(x)\to-\infty$ as $n\to\infty$. To do so, we calculate
     \begin{equation}
	  j_{\psi_n^-}(x)
	  =
	  \left(\begin{array}{c}\overline\alpha\\\overline\beta\end{array}\right)^t
	  \big(j_{\chi}(x)\cdot \openone + n A_n \big)\left(\begin{array}{c}\alpha\\\beta\end{array}\right),
     \end{equation}
     where
     \begin{equation}
       A_n:= \begin{pmatrix}
		 0 & e^{inx}(\frac{j_{\chi}(x)}{n}+\frac{|\chi(x)|^2}{2})\\
		 e^{-inx}(\frac{j_{\chi}(x)}{n}+\frac{|\chi(x)|^2}{2}) & |\chi(x)|^2
	     \end{pmatrix} .
     \end{equation}
     The $(2\times 2)$ matrix $A_n$ is hermitian, has trace $|\chi(x)|^2$, and determinant $\det A_n\to-|\chi(x)|^4/4<0$ as $n\to\infty$. Thus the eigenvalues $\lambda_\pm(n)$ of $A_n$ converge to $\lambda_\pm(n)\to (1\pm \sqrt{2}) |\chi(x)|^2/2$. Choosing $\alpha,\beta$ as the coordinates of an eigenvector with the negative eigenvalue $(1- \sqrt{2}) |\chi(x)|^2/2$ then results in $j_{\psi_n^-}(x)\to-\infty$ as $n\to\infty$ because of the explicit prefactor $n$ in front of $A_n$.
\end{proof}

The superpositions constructed in the second part of the proof are examples of states in which backflow occurs (``backflow states''). For other examples, see \cite{YearsleyHalliwellHartshornWhitby:2012,HalliwellGillmanLennonPatelRamirez:2013}.

\bigskip

On the mathematical side, the observable $J(x)$, defined as
\begin{align}
     \hscalar{\psi}{J(x)\psi}:=j_\psi(x)
\end{align}
exists only as a quadratic form, and is unbounded above and below on $E_+L^2(\Rl)$ by the results above. This quadratic form will be our main tool in studying the dependence of the probability current density on the wave function. In particular, $J$ encodes bounds on spatial averages of the backflow against (positive) Schwartz-class test functions $f\in\Ss(\Rl)$, written as
\begin{align}
     \hscalar{\psi}{J(f)\psi}
     :=
     \int dx\,f(x)\,j_\psi(x)\,.
\end{align}
It is readily checked that $J(f)$ exists as an (unbounded) operator, hermitian for real $f$, and can be expressed in terms of the position and momentum operators as
\begin{align}\label{eq:Jf-Schrodinger}
     J(f)=\frac{1}{2}\left(Pf(X)+f(X)P\right).
\end{align}
The fact that backflow exists is reflected in the fact that $E_+J(f)E_+$, the averaged current evaluated in right-moving states, is not positive. To formulate this concisely, let us denote by 
\begin{align}
     \spinf{ A}
     :=\inf_{\|\psi\|=1}\langle\psi,A\psi\rangle
     \in[-\infty,\infty)
\end{align}
the bottom of the spectrum of a hermitian operator $A$, i.e., the infimum of all its generalized eigenvalues. Then the maximal amount of backflow, spatially averaged with $f$, is defined as
\begin{align}\label{eq:beta0}
     \beta_0(f):= \spinf{ E_+J(f)E_+ }\,.
\end{align}

In the following Theorem~\ref{thm:backflow-bounds-basic}, we summarize three properties $i)$-$iii)$ of $J(f)$ that are relevant to our investigation. Part $i)$ is concerned with the {\em existence} of backflow: By Prop.~\ref{proposition:jx-unbounded}, we can pick positive test functions $f$ such that $E_+J(f)E_+$ is not positive. Below we give a stronger argument, showing that $\beta_0(f)<0$ for {\em each real} $f\neq0$. 

Having settled the existence of backflow, the next question concerns the strength of this effect. In part $ii)$, we remark that $E_+J(f)E_+$ is unbounded {\em above} for positive $f$, just as $E_+J(x)E_+$. This is intuitively clear, saying that there is no restriction on probability flow to the right for right-moving waves, and follows in a similar manner as in the first part of Prop.~\ref{proposition:jx-unbounded}. 

A more delicate question is whether there exist {\em lower} bounds on the spectrum of the smeared probability current $E_+J(f)E_+$, i.e., whether $\beta_0(f)>-\infty$. In fact, $E_+J(f)E_+$ is bounded below, in contrast to $E_+J(x)E_+$. Part $iii)$ recalls a result proven by Eveson, Fewster, and Verch \cite{EvesonFewsterVerch:2003} in this context.

\begin{theorem}\label{thm:backflow-bounds-basic}{\bf (Existence and boundedness of spatially averaged backflow)}
	\begin{enumerate}
		\item\label{it:backflowallf} For any real $f\neq0$, the smeared probability flow in right-moving states, $E_+J(f)E_+$, is not positive, $\beta_0(f)<0$.
		\item\label{it:backflowubabove} Let $f>0$. Then there is no finite upper bound on $E_+J(f)E_+$.
		\item\label{it:backflowbdbelow} \cite{EvesonFewsterVerch:2003} Let $f>0$. Then $E_+J(f)E_+$ is bounded below, i.e., $\beta_0(f)>-\infty$. For test functions of the form $f=g^2$ for some real $g\in\Ss(\Rl)$, one has
		\begin{align}\label{eq:bound-for-square}
			\beta_0(g^2)
			\geq
			-\frac{1}{8\pi}\int dx\,|g'(x)|^2>-\infty\,.
		\end{align}
	\end{enumerate}
\end{theorem}
\begin{proof}
     \emph{i)} The operator $E_+J(f)E_+$ defines an integral operator on $L^2(\Rl_+,dp)$. In view of \eqref{eq:Jf-Schrodinger}, its integral kernel is 
     \begin{align}\label{eq:basic-kernel}
	  K_f(p,q)=\frac{p+q}{2\sqrt{2\pi}}\,\tilde f(p-q),\qquad p,q\geq0\,.
     \end{align}
     If $E_+J(f)E_+$ is positive, then for any $p,q>0$, the hermitean matrix
     \begin{align}
	  \left(
		    \begin{array}{cc}
			 K_f(p,p) & K_f(p,q)\\
			 K_f(q,p) & K_f(q,q)
		    \end{array}
	  \right)
     \end{align}
     has only non-negative eigenvalues, and in particular a non-negative determinant
     \begin{equation}
     \begin{aligned}
	  0 &\leq K_f(p,p)K_f(q,q)-|K_f(p,q)|^2
	  \\ &=
	  \frac{pq}{2\pi}\,|\tilde f(0)|^2-\frac{(p+q)^2}{8\pi}\,|\tilde f(p-q)|^2\,.
     \end{aligned}
     \end{equation}
     This implies
     \begin{align}
	  |\tilde f(p-q)|
	  \leq
	  \frac{2\,\sqrt{p\,q}}{p+q}\,|\tilde f(0)|\,.
     \end{align}
     Now taking $p\to0$ at fixed $q>0$ shows that $|\tilde f(-q)|=0$ for all $q>0$. But since $f$ is real, $\tilde f(-q)=\overline{\tilde f(q)}$, so that $\tilde f(q)=0$ for each $q\neq0$. As the test function $f$ is continuous, this implies that $f=0$ has to vanish altogether. So we conclude that for any real $f\neq0$, the operator $E_+J(f)E_+$ is not positive.

     \smallskip

     \emph{ii)} Similar to the proof of the first part of Prop.~\ref{proposition:jx-unbounded}, we take a normalized and right-moving $E_+\psi=\psi\in L^2(\Rl)$, and define shifted momentum space wave functions, $\tilde\psi_n(p):=\tilde\psi(p-n)$, where $n>0$. Then also $\psi_n$ is normalized, $E_+\psi_n=\psi_n$, and has the expectation value (cf. \eqref{eq:j-shift} integrated against $f(x)$ over $x$)
     \begin{equation}
     \begin{aligned}
	  \hscalar{\psi_n}{E_+J(f)E_+\psi_n}
	  &= 
	  \hscalar{\psi}{E_+J(f)E_+\psi}
	  \\ &\quad +
	  n\,\int dx\,f(x)\,|\psi(x)|^2\,.
     \end{aligned}
     \end{equation}
     For $f>0$, it is clear that we can choose $\psi$ in such a way that the last integral is not zero. In that case, $\langle \psi_n,E_+J(f)E_+\psi_n\rangle\to\infty$ as $n\to\infty$, showing that there is no finite upper bound on the spectrum of $E_+J(f)E_+$.
\end{proof}

The spectrum of the operator $E_+J(f)E_+$ cannot be determined explicitly, and analytic methods are restricted to providing bounds on the backflow effect. (Numerical results will be presented in Section \ref{sec:examples}.) 

The significance of the mild additional assumption $f=g^2$ in Thm.~\ref{thm:backflow-bounds-basic}\ref{it:backflowbdbelow} is due to the fact that in this case, $J(f)$ takes a simpler form. Namely, in view of Heisenberg's commutation relation $[X,P]=i$, one obtains the more symmetric formula $J(g^2)=g(X)Pg(X)$. With the help of the spectral projections $E_\pm$ of $P$, one may then write $J(g^2)$ as a difference of two positive operators,
\begin{equation}
\begin{aligned}\label{eq:J-split}
     J(g^2) & =J_+(g^2)-J_-(g^2)
     \\
    \text{with } J_\pm(g^2) &:= \pm g(X)PE_\pm g(X)\,,
\end{aligned}
\end{equation}
which yields the estimate $\beta_0(g^2)\geq-\|E_+J_-(g^2)E_+\|$. The inequality \eqref{eq:bound-for-square} can then be established by estimating this norm \cite{EvesonFewsterVerch:2003}. 

From \eqref{eq:J-split} we see that the negative part $-J_-(g^2)$ (without restriction to right-moving waves) is unbounded, but the unboundedness occurs, roughly speaking, only at high momentum. This will be important in our subsequent investigation of backflow in scattering situations.

\section{Backflow and scattering} \label{sec:scattering}

Let us now consider a quantum mechanical system as before, but with non-trivial interaction given by a time-independent external potential $V$, so that the Hamiltonian is of the form $H = \frac{1}{2}P^2+ V(X)$.

In this situation the time evolution does no longer leave the space of right-movers $E_+L^2(\Rl)$ invariant. Hence, what constitutes a particle that ``travels to the right'' is less clear.  As a substitute, we propose to look at \emph{asymptotic} momentum distributions in the sense of scattering theory, that is, states whose incoming asymptote is a right-mover. This space \emph{is} invariant under the time evolution; it describes particles scattering ``from the left'' onto the potential. The connection between an asymptotic state $\psi$ and the ``interacting state" $\Omega_V\psi$ is given by the incoming \emph{M{\o}ller operator}
\begin{equation}\label{omega:time}
\Omega_V := \operatorname*{s-lim}_{t \rightarrow - \infty} e^{iHt} e^{-iH_0 t}\,,
\end{equation}
where $\operatorname{s-lim}$ denotes the strong operator limit. We remark that, although $\Omega_V$ is not unitary in the presence of bound states, we still have $\lVert \Omega_V \rVert = 1$.

We will now look at the averaged probability current $J(f)$ in states with right-moving asymptote $\psi=E_+\psi$. That is, we consider the ``asymptotic current operator'' 
$E_+ \Omega_V^\ast J(f)\Omega_V E_+$
and investigate its spectral properties -- whether it is unbounded above (unlimited forward flow), 
bounded below (limited backflow), and how to estimate
\begin{equation}\label{main}
\beta_V(f) := \spinf{ E_+ \Omega_V^\ast J(f) \Omega_V E_+}\,,
\end{equation}
the ``asymptotic backflow constant''. For $V=0$, one has $\Omega_V=1$ and hence recovers the previously discussed $\beta_0(f)$, see \eqref{eq:beta0}. Our main result will be that also for short-range potentials $V \neq 0$, one has $\beta_V(f) > -\infty$ for any non-negative test function $f$.

In order to be able to do scattering theory, we work with potentials $V(x)$ that vanish sufficiently fast as $x \to \pm\infty$. Specifically, we consider real-valued  potentials $V$ for which the norm
\begin{equation}
\lVert  V \rVert_{1+} := \int dx\, (1 + \lvert x\rvert) \lvert  V(x)\rvert < \infty
\end{equation}
is finite; we refer to this class as $L^{1+}(\mathbb{R})$. In this case  solutions of the time-independent Schr\"odinger equation on the line,
\begin{equation}\label{schroe}
(-\partial_x^2 + 2V(x) - k^2)\psi(x) =0,\quad k \in \Rl\,,
\end{equation}
are scattering states. In the stationary picture of scattering theory, of particular interest are the solutions $x \mapsto \varphi_k(x)$, $k>0$, of \eqref{schroe} with the asymptotics 
\begin{equation}\label{eq:asymp}
\varphi_k(x) = \begin{cases} 
T_V(k) e^{ikx} + o(1) & \text{for } x \gg 0, \\
 e^{ikx} +R_V(k)e^{-ikx} + o(1)& \text{for }x \ll 0\,,
\end{cases}
\end{equation} 
where $R_V(k)$ and $T_V(k)$ denote the reflection and transmission coefficients of the potential $V$, respectively. Let us recall the basic results of scattering theory in this context; see, e.g., \cite{Yafaev:analytic}.

\begin{lemma}
Let $V\in L^{1+}(\mathbb{R})$. Then the operator $\Omega_V$ defined in \eqref{omega:time} exists. Further, the solution $x \mapsto \varphi_k(x)$ ($k>0$) of \eqref{schroe} with the asymptotics \eqref{eq:asymp} exists and is unique, and for any $\tilde\psi \in C_0^\infty(\mathbb{R})$, 
\begin{equation}\label{eq:omegakernel}
(\Omega_V E_+ \psi)(x) = \frac{1}{\sqrt{2\pi}}\int_0^\infty dk\, \varphi_k(x) \tilde\psi(k)\,.
\end{equation} 
\end{lemma}

Existence and uniqueness of $\varphi_k$ are a consequence of \cite[Chapter~5, Lemma~1.1]{Yafaev:analytic}. The proof of existence of $\Omega_V$ under a weaker assumption on $V$ can be found in \cite[Chapter~5, Theorem 1.12]{Yafaev:analytic}, as well as the relation \eqref{eq:omegakernel}.
Note that our $\varphi_k(x)$ is denoted $\psi_1(x,k)$ there.

By Eq.~\eqref{eq:omegakernel}, the expectation values of the asymptotic current operator are
\begin{equation}\label{eq:jvkernel}
\begin{aligned}
   \hscalar{\psi}{  E_+ & \Omega_V^\ast J(f) \Omega_V E_+  \psi} = \int dx f(x) \\
   \times & \int_0^\infty dp \int_0^\infty dq \, \overline{ \tilde \psi(p)} \, K_{V}(p,q,x) \, \tilde\psi(q)\,,
\end{aligned}
\end{equation}
where
\begin{equation}
 K_{V}(p,q,x) = \frac{i}{4\pi} \Big( \overline{\partial_x\varphi_p(x)} \varphi_q(x) - \overline{\varphi_p(x)} \partial_x \varphi_q(x) \Big)\,.
\end{equation}

For estimating this operator, we will rely on the following pointwise bounds on $\varphi_k$ and $K_V$ which relate them to their spatial asymptotics (the transmitted wave).
\begin{lemma}{\bf \cite{DT:scattering}}\label{lemma:m}
     Let $V \in L^{1+}(\mathbb{R})$. There exist constants $c_V,c_V',c_V'',c_V''' >0$ such that for all $x \in \mathbb{R}$ and $k>0$,
     \begin{align}
	  \lvert \varphi_k(x) \rvert  &\leq c_V (1 + \lvert x \rvert), \label{eq:psibound} \\
	  \lvert \varphi_k(x) - e^{ikx} \rvert  &\leq c_V' \frac{1 + \lvert x \rvert}{1+k}, \label{eq:psiminusfree} \\
	  \lvert \partial_x \varphi_k(x) - ik \varphi_k(x) \rvert  &\leq c_V'' \frac{1}{1+k}, \label{eq:dpsiminusfree} \\
	  \lvert K_V(p,q,x) - \frac{p+q}{4\pi} \overline{\varphi_p(x)} \varphi_q(x) \rvert  &\leq c_V''' (1 + \lvert x \rvert )\,. \label{eq:Kminusfree}
     \end{align}
\end{lemma}

The estimates \eqref{eq:psiminusfree} and \eqref{eq:dpsiminusfree} can be deduced from \cite[Sec.~2, Lemma 1]{DT:scattering}, noting that the function $m(x,k)$ there corresponds to our $\varphi_k(x)e^{-ikx}/T_V(k)$, and that $\lvert T_V(k) \rvert \leq 1$ and $T_V(k) = 1 + O(1/k)$ for large $k$ \cite[Sec.~2, Theorem 1]{DT:scattering}. Eqs.~\eqref{eq:psibound} and \eqref{eq:Kminusfree} are consequences of \eqref{eq:psiminusfree} and \eqref{eq:dpsiminusfree}. The constants $c_V$ etc.\ can in principle be deduced from \cite{DT:scattering} as functions of $V$, but these bounds are not optimal and we will not need them in the following. In Sec.~\ref{sec:examples}, we will consider specific examples for $V$ where $\varphi_k$ and an optimal $c_V$ can be computed.

\bigskip

With this information at hand, we first investigate unboundedness from above and the existence of negative parts of the spectrum, generalizing the results of the free situation.
\begin{theorem}\label{thm:existence-of-backflow-scattering}{\bf(Existence of backflow in scattering situations)}\\
   Let $V \in L^{1+}(\mathbb{R})$. 
   \begin{enumerate}
    \item \label{it:unbAboveV} For every $f>0$, there is no finite upper bound on the operator $E_+ \Omega_V^\ast J(f) \Omega_V E_+ $.
    \item \label{it:unbBelowV} For every $x \in \Rl$, there is a sequence of normalized right-movers $\psi_n=E_+ \psi_n$ such that $\hscalar{\psi_n}{\Omega_V^\ast J(x) \Omega_V \psi_n} \to -\infty$ as $n \to \infty$.

   \end{enumerate}

\end{theorem}

Point \ref{it:unbBelowV} implies in particular that $\beta_V(f) < 0$ for certain positive $f$, that is, averaged backflow exists in all scattering situations. In the free case, we were able to show that $\beta_0(f) < 0$ for \emph{all} positive $f$  (Thm.~\ref{thm:backflow-bounds-basic}\ref{it:backflowallf}), but we currently have no proof of the analogous statement for $\beta_V$ in the interacting situation.

\begin{proof}
   \ref{it:unbAboveV} As in the proof of Thm.~\ref{thm:backflow-bounds-basic}\ref{it:backflowubabove}, we pick a right-mover $\psi$ and shift it to higher and higher momentum, $\tilde\psi_n(p):=\tilde\psi(p-n)$. Then all $\psi_n$ are right-movers and $\|\psi_n\|=\|\psi\|$ for all $n$.
   In view of the unboundedness of $\hscalar{\psi_n}{ E_+J(f)E_+ \psi_n}$ from above (Thm.~\ref{thm:backflow-bounds-basic}\ref{it:backflowubabove}), it suffices to show that $\hscalar{\psi_n}{ (\Omega_V^\ast J(f) \Omega_V - J(f)) \psi_n}$ is bounded as $n \to \infty$. In fact, from \eqref{eq:jvkernel} we have
   \begin{equation}\label{eq:jv-unbounded-above}
 \begin{aligned}
  \hscalar{\psi_n}{ &(\Omega_V^\ast J(f) \Omega_V - J(f)) \psi_n} = \int dx\,f(x) \int dp\,dq\\
  \times & \Big\{ \overline{\tilde\psi_n (p)} \tilde \psi_n(q) \, 
         \big\lbrack K_V(p,q,x) - \frac{p+q}{4\pi} \overline{\varphi_p(x)} \varphi_q(x)  \big\rbrack\\
   +&  \overline{\tilde\psi (p)} \tilde \psi(q) \,
        \frac{p+q+2n}{4\pi} \big\lbrack\, \overline{\varphi_{p+n}(x)} \varphi_{q+n}(x) - e^{i(q-p)x} \big\rbrack \Big\}.
 \end{aligned}
 \end{equation}
 Due to Eq.~\eqref{eq:Kminusfree}, and since the norms $\|\tilde\psi_n\|_1$ are independent of $n$, the first summand gives a bounded contribution; Eqs.~\eqref{eq:psibound} and \eqref{eq:psiminusfree} yield the same for the second summand.
 This proves (i).---%
   Similarly for \ref{it:unbBelowV}, with $\psi_n$ being the sequence $\psi_n^-$ from Prop.~\ref{proposition:jx-unbounded}, it suffices to show that $\hscalar{\psi_n}{ (\Omega_V^\ast J(x) \Omega_V - J(x)) \psi_n}$ is bounded,
    which follows with the same technique as in \ref{it:unbAboveV}, using a suitable choice for $\chi$.
\end{proof}

Let us now turn to our main result: the boundedness of the backflow constant $\beta_V(f)$ for every fixed non-negative test function $f$. 
To that end, we use $E_+ + E_- = 1$ and split the expression $E_+ \Omega_V^\ast J(f) \Omega_V E_+$ into several terms. Here the product $E_- \Omega_V E_+$ can be rewritten as $E_- (\Omega_V - T_V)E_+$, where $T_V$ acts by  multiplication with $ T_V(k)$ in momentum space. (Hence, $T_V$ commutes with the spectral projection $E_\pm$ and $E_- T_V E_+ = E_- E_+ T_V = 0$.) Additionally, with the aim of controlling the unboundedness of $J(f)$, we multiply it on one side with the operator $(i+P)^{-1}$. We obtain
\begin{equation}
\begin{aligned}
E_+ &\Omega_V^\ast J(f)\Omega_V E_+ = \\
&  E_+ \Omega_V^\ast E_+ J(f) E_+ \Omega_V E_+   \\
+ &E_+ \Omega_V^\ast E_+ J(f) (i+P)^{-1} E_- (i+P)(\Omega_V - T_V) E_+  \\
+ &E_+ (\Omega_V^\ast - T_V^\ast)(-i+P)E_- (-i +P)^{-1}J(f)  \Omega_V E_+\,.
\end{aligned}
\end{equation}
The first term on the right-hand side is bounded below by $\beta_0(f)$, as known from Thm.~\ref{thm:backflow-bounds-basic}\ref{it:backflowbdbelow} and $\|E_+\|=\|\Omega_V\|=1$. This yields
\begin{equation}\label{sum}
\begin{aligned}
&E_+ \Omega_V^\ast J(f)\Omega_V E_+  
\\
&\geq \beta_0(f) 
-  2\| J(f) (i+P)^{-1} \|  \, \| (i+P)(\Omega_V - T_V) E_+ \|  
\\
&\geq \beta_0(f) 
-  2 \| J(f) (i+P)^{-1} \|  ( 2 +  \| P (\Omega_V - T_V) E_+ \| )\,.
\end{aligned}
\end{equation}
It remains to show that the two norms on the right-hand side are finite. From \eqref{eq:Jf-Schrodinger}, one finds  $J(f) = f(X)P -\frac{i}{2}f'(X)$ and hence
\begin{equation} \label{eq:jp}
\lVert  J(f)(i+P)^{-1} \rVert \leq \lVert f \rVert_{\infty} + \frac{1}{2}\lVert f' \rVert_{\infty}\,,
\end{equation} 
where $\|f\|_\infty=\sup_{x\in\Rl}|f(x)|$ denotes the supremum norm.

In order to estimate $\|P (\Omega_V - T_V) E_+ \|$, it will be helpful to express the Schr\"odinger equation in suitable integral form (Lippman-Schwinger equation). It is easily checked that
\begin{equation}
G_k(x) := \frac{\sin(kx)}{k} \Theta(x)
\end{equation}
is a Green's function for the free Schr\"odinger equation, i.e., $-G''_k(x) = k^2 \cdot G_k(x) - \delta(x)$ in the sense of distributions.
The solution $\varphi_k$ is then uniquely determined by
\begin{equation}\label{lippmanschwinger}
\varphi_k(x) = T_V(k)e^{ikx} + \int dy\, 2V(y)G_k(y - x)\varphi_k(y)\,.
\end{equation}
With this information, we now prove the following proposition.
\begin{proposition}\label{proposition:pomegat}
     Let $V \in L^{1+}(\mathbb{R})$. Then 
     \begin{equation}
     \lVert P(\Omega_V - T_V) E_+ \rVert  \leq 2 c_V \lVert V \rVert_{1+}
     \end{equation}
     with the constant $c_V$ from Lemma~\ref{lemma:m}.
\end{proposition}

\begin{proof}
Let $\tilde\psi,\tilde\xi \in C_0^\infty(\Rl)$ with $E_+ \psi = \psi$. Using \eqref{eq:omegakernel}, we can write
\begin{equation}\label{Fminus}
\begin{aligned}
\langle \xi,  & P (\Omega_V - T_V) \psi \rangle 
\\ &=\frac{i }{\sqrt{2\pi}} \int  dx \, \overline{\xi'(x)} \int_0^\infty  dk\, \big(  \varphi_k(x) - T_V(k)e^{ikx} \big) \tilde\psi(k)\,.
\end{aligned}
\end{equation}
In view of the Lippman-Schwinger equation \eqref{lippmanschwinger}, we may rewrite the above expression as
\begin{equation}\label{lippgreen}
\begin{aligned}
\hscalar{ \xi}{ P& (\Omega_V - T_V) \psi  }
 =\frac{i}{\sqrt{2\pi}} \int dx \, \overline{\xi'(x)} \int_0^\infty  dk \int dy\\
 &\qquad \times 2 V(y) G_k(y-x) \varphi_k(y) \tilde\psi(k) \\
& = \frac{2i}{\sqrt{2\pi}} \int dx\, dy  \int_0^\infty  dk \, \overline{\xi(x)} V(y) 
\\ &\qquad \times \cos(k(y - x))\Theta(y - x)  \varphi_k (y) \tilde\psi(k)\,,
\end{aligned}
\end{equation}
where we have used Fubini's theorem and integrated by parts. To estimate this integral, let us introduce the multiplication and integral operators $(M_{y}\psi)(k) := \varphi_k (y)\cdot \tilde\psi(k)$ and $(I_{y} \tilde\psi)(x) := \Theta(y - x) \int_{0}^\infty dk\,  \cos(k(y - x)) \tilde  \psi(k)$. Then, by Lemma~\ref{lemma:m}, we have $\lVert M_{y}  \rVert \leq c_V (1 + \lvert y \rvert)$ for all $y \in \mathbb{R}$. The integral operator $I_{y}$ consists of a projection onto the even and positive momentum part of $\psi$, a multiple of the Fourier transform, multiplication by the Heaviside function and the coordinate change $x \rightarrow y - x$. This implies $\lVert  I_{y} \rVert \leq \sqrt{2\pi}$ for all $y \in \mathbb{R}$, which then yields
\begin{equation}
\begin{aligned}
| \langle \xi,  P(\Omega_V &- T_V) \psi \rangle |   
\\
&\leq \frac{2 \lVert \xi \rVert \lVert \psi \rVert }{\sqrt{2\pi}}  \int dy\, \lvert V(y)  \rvert \lVert I_{y} \rVert \lVert M_{y} \rVert 
\\
&\leq 2 c_V \lVert V \rVert_{1+} \cdot \lVert \xi \rVert \lVert \psi \rVert\,.
\end{aligned}
\end{equation}
As $\xi$ and $\psi$ were taken from a dense subspace of $L^2(\mathbb{R})$ and $E_+ L^2(\mathbb{R})$, respectively, this finishes the proof.
\end{proof}
Combining Eqs.~\eqref{sum} and \eqref{eq:jp} and Prop.~\ref{proposition:pomegat}, we arrive at the following result.
\begin{theorem}\label{thm:scattering-bounds}{\bf(Boundedness of backflow in scattering situations)}\\
For any potential $V \in L^{1+}(\mathbb{R})$ and any non-negative $f$, there exists a lower bound on the backflow:
\begin{equation}\label{finalbound}
\beta_V(f) \geq \beta_0(f) -(2\lVert f \rVert_{\infty} + \lVert f' \rVert_{\infty} ) (2 +2 c_V \lVert V \rVert_{1+}) > -\infty\,.
\end{equation}
Here $c_V$ is the constant from Lemma~\ref{lemma:m}.
\end{theorem}
Thus the backflow effect is limited for short-range potentials in the class $L^{1+}(\mathbb{R})$. Here the falloff of $V$ at large $|x|$ was important for our argument, as otherwise several of the integrals considered would not be finite. On the other hand, the specified behavior at short distances, namely, that $V$ is a locally integrable function, is not strictly required, and it is not too hard to generalize the result for situations with delta-like point interactions.

Let us illustrate this for a finite sum of delta potentials, $V(x)=\sum_j \lambda_j \delta(x-x_j)$. In this situation, the M\o{}ller operator $\Omega_V$ still exists and has the form \eqref{eq:omegakernel} \cite{Duchene:2011}; the solutions $\varphi_k$ of the Schr\"odinger equation, which are piecewise a superposition of two plane waves, fulfill the bound $|\varphi_k(x)| \leq c_V$ with some constant $c_V$, and they satisfy the Lippman-Schwinger equation \eqref{lippmanschwinger} in the sense of distributions. We can then follow a similar argument as above, and conclude
\begin{equation}
\begin{aligned}
&\big(P  (\Omega_V - T_V)E_+ \psi \big)(x) = 
- \sum_j \frac{2 \lambda_j}{\sqrt{2\pi}i} 
\\ &\times  \int_0^\infty dk \cos(k(x-x_j)) \Theta(x_j-x) \varphi_k(x_j) \tilde\psi(k)\,. 
\end{aligned}
\end{equation}
The multiplication and integral operators $M_{y}$ and $I_{y}$ assume the expressions $(M_j \psi)(k) := \varphi_k(x_j) \cdot \tilde\psi(k)$ and $(I_j \tilde\psi)(x) := \Theta(x_j-x) \int_0^\infty dk\, \cos(k(x-x_j)) \tilde\psi(k)$, and we obtain $\lVert M_j \rVert \leq c_V$ and $\lVert I_j \rVert \leq \sqrt{2\pi}$. Hence, we have
\begin{equation}
\begin{aligned}
\lVert P(\Omega_V - T_V)E_+ \rVert 
&\leq  \frac{2 }{\sqrt{2\pi}}\sum_j\lvert \lambda_j \rvert \lVert I_j \rVert \cdot \lVert  M_j \rVert 
\\
&\leq 2 c_V \sum_j \lvert \lambda_j \rvert\,, 
\end{aligned}
\end{equation} 
and \eqref{finalbound} becomes
\begin{equation}\label{delta}
\beta_V(f) \geq \beta_0(f) -(2\lVert f \rVert_{\infty} + \lVert f' \rVert_{\infty} )\cdot (2 + 2 c_V \sum_j \lvert  \lambda_j \rvert )\,.
\end{equation}
Thus our lower bounds also hold for point interactions.

It may seem surprising at first that $\beta_V(f) > -\infty$ even in the presence of a reflecting potential. One can understand this qualitatively as follows. As discussed in Sec.~\ref{sec:free}, unboundedness below of the smeared current $J(f)$ is a high-momentum effect, with contributions growing like $O(p)$. In our class of potentials, however, the reflection coefficient $R(p)$ approaches zero at high momenta, at least like $O(1/p)$ \cite{DT:scattering}. This exactly compensates the high-momentum divergence of $J(f)$ and leads to limits on backflow like in the free case.

\section{Examples} \label{sec:examples}

In the previous section, we have shown that the backflow effect is limited in generic scattering situations, i.e., that the operator $E_+ \Omega_V^\ast J(f) \Omega_V E_+$ is bounded below, but little was said about the actual value of the bound. We now want to investigate this further in specific examples of potentials, both starting from the analytic estimates in Sec.~\ref{sec:scattering} and with numeric methods.
We will investigate the asymptotic backflow in the following potentials:
\begin{enumerate}
 \item the zero potential ($V=0$, free particle) as a reference;
 \item a single delta potential, $V(x)=\lambda \delta(x)$, both in the attractive ($\lambda < 0$) and repulsive ($\lambda>0$) case;
 \item the rectangular potential, $V(x) = \lambda \Theta(1-x)\Theta(1+x)$, again repulsive or attractive;
 \item the P\"oschl-Teller potential \cite{Poeschl:oscillator}, given by
 \begin{equation}
     V(x) = - \frac{\mu(\mu+1)}{2 \cosh^2 x} \quad \text{with } \mu>0\,.
 \end{equation}
 This has the particular property \cite{Lekner:reflectionless} that for integer $\mu$, the potential becomes \emph{reflectionless,} that is, the reflection coefficient vanishes for all momenta. This will allow us to specifically investigate the influence of reflection on the backflow.

\end{enumerate}

For simplicity, we will always take our smearing function $f$ to be a Gaussian with a fixed width $\sigma$ and center $x_0$, i.e.,
\begin{equation}
f(x) = \frac{1}{\sigma \sqrt{2\pi}} \exp \left(- \frac{(x - x_0)^2}{2\sigma^2}  \right). 
\end{equation}

Let us first concretize our analytic estimates. For the free particle, \eqref{eq:bound-for-square} gives
\begin{equation}\label{eq:freeest}
 \beta_0(f) \geq - \frac{1}{32 \pi \sigma^2}\,.
\end{equation}
For the single delta potential, the solution $\varphi_k$ equals $e^{ikx} + R_V (k)e^{-ikx}$ for $x <0$, and $T_V (k)e^{ikx}$ for $x>0$; thus we can choose $c_V=2$, and \eqref{delta} yields
\begin{equation}\label{eq:single-delta-with-gauss}
  \beta_V(f) \geq  -\frac{1}{32\pi \sigma^2} - \big( \frac{2}{\sigma \sqrt{2\pi}} + \frac{1}{\sigma^2 \sqrt{2\pi e}} \big) \cdot (2 + 4|\lambda|)\,.
\end{equation}
In the P\"oschl-Teller case, let us restrict ourselves to the case $\mu = 1$. Here the solution of the Schr\"odinger equation can be explicitly written as
$\varphi_k(x) = e^{ikx} (k +i\tanh x)/(k -i)$, so that $|\varphi_k| \leq 1$; that is, in Lemma~\ref{lemma:m}, we can set $c_V=1$, and Theorem~\ref{thm:scattering-bounds} yields
\begin{equation}\label{poeschl1}
\beta_V(f) \geq -\frac{1}{32\pi \sigma^2} - \big( \frac{2}{\sigma \sqrt{2\pi}} + \frac{1}{\sigma^2 \sqrt{2\pi e}} \big) \cdot (6 + 4\ln 2)
\end{equation}
for this potential. The rectangular potential can be treated with similar methods, using bounds on the explicitly known solution $\varphi_k$, though we skip the details here.

All these are only lower bounds to the backflow $\beta_V(f)$ -- and as we shall see below, they are quite rough estimates. The actual value of $\beta_V(f)$ is not accessible to our explicit computations, even in concrete examples of the potential; it can be obtained only by numerical approximation. In view of Eq.~\eqref{eq:jvkernel}, $\beta_V(f)$ is the lowest spectral value of the integral operator on $L^2(\mathbb{R}_+)$ with kernel
\begin{equation}\label{eq:kvf}
\begin{aligned}
   K_{V,f}(p,q) = & \frac{i}{4\pi} \int  dx \, f(x) \\
   \times& \Big( \overline{\partial_x\varphi_p(x)} \varphi_q(x) - \overline{\varphi_p(x)} \partial_x \varphi_q(x) \Big)\,.
 \end{aligned}
\end{equation}
Knowing the solutions $\varphi_p(x)$ for a given potential $V$, we can use a discretization of the wave functions in momentum space $L^2(\mathbb{R}_+)$ to approximate the integral operator by a hermitean matrix; the lowest eigenvalue of this matrix is then an approximation for $\beta_V(f)$, and we can obtain an approximation of the corresponding eigenfunction as well. Details of the numerical method are described in the Appendix. Let us just mention at this point that this involves an upper cutoff $p_\mathrm{max}$ for the momentum of the wave functions, and a number $n$ of discretization steps; these parameters will enter the approximations below.

\begin{figure}
     \centering
     \begin{subfigure}[t]{0.45\textwidth}
      \includegraphics[width=\textwidth]{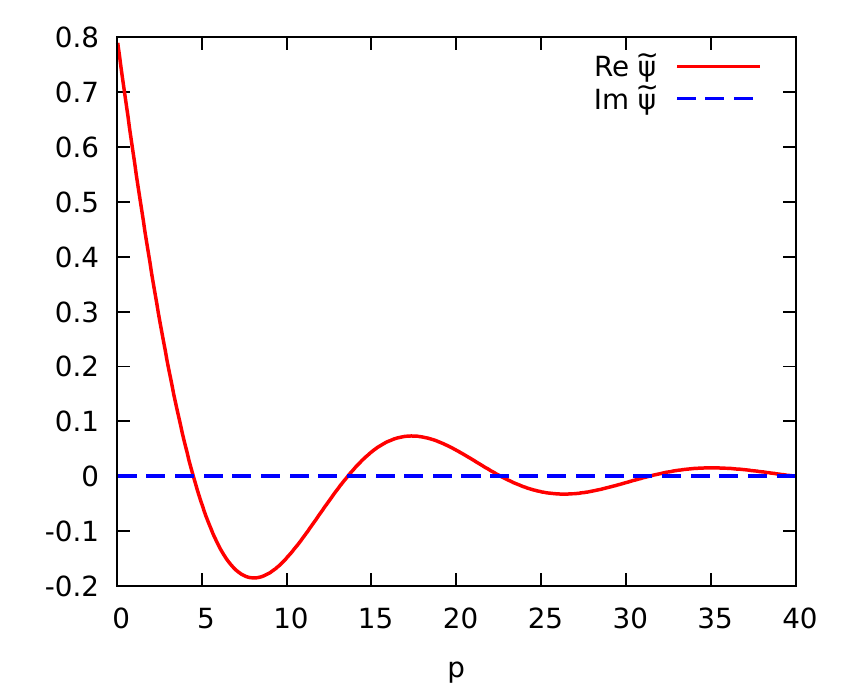}
      \caption{Zero potential}
      \label{fig:eigenvec-free} 
     \end{subfigure}%
     \quad
     \begin{subfigure}[t]{0.45\textwidth}
      \includegraphics[width=\textwidth]{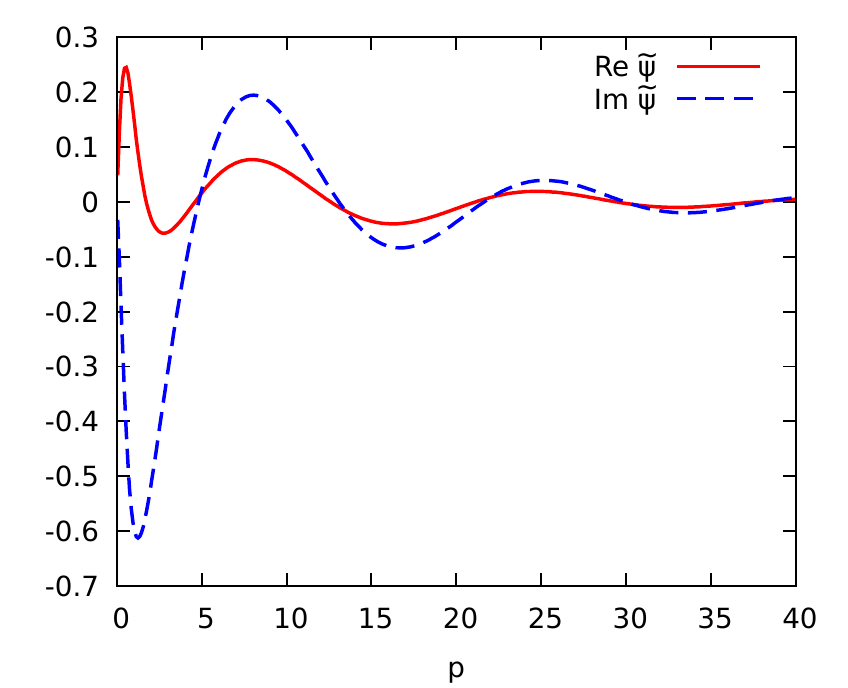}
      \caption{Delta potential, $\lambda=1$}
      \label{fig:eigenvec-delta} 
     \end{subfigure}%
     \caption{Lowest eigenvector of the asymptotic current operator. Parameters: $n=2000$, $p_\mathrm{max}=200$, $x_0=0$, $\sigma = 0.1$.}
     \label{fig:eigenvec}
\end{figure}

We will now analyze the dependence of $\beta_V(f)$ and of the corresponding lowest eigenvector on parameters of the system, both with numeric and analytic methods.

\subsection{Eigenfunctions}
To start, let us look at the momentum space wave function of the eigenvector for the lowest eigenvalue. In the free particle case, the numerically obtained eigenfunction is real-valued and shown in Fig.~\ref{fig:eigenvec-free}. The corresponding eigenvalue is $\beta_0(f)\approx -0.241$, while the estimate in \eqref{eq:freeest} gives $\beta_0(f) \geq -0.995$, almost an order of magnitude from the numerical result. (Here and in the following, the numeric values for $\beta_0(f)$ and $\beta_V(f)$ need to be read in units of $\hbar/m \ell^2$, where $\ell$ was the chosen unit of length.)
The oscillating graph confirms that, as expected from the analytic derivation, backflow is an interference effect between low-momentum and high-momentum portions of the wave function.  Also, the eigenfunction decays quite rapidly at large momenta, showing that our cutoff $p_\mathrm{max}$ in momentum space is at least self-consistent. For a delta potential, Fig.~\ref{fig:eigenvec-delta}, which we take here as a simple example of the interacting situation, the eigenvector has similar qualitative features.

\begin{figure}
     \centering
     \begin{subfigure}[t]{0.45\textwidth}
             \includegraphics[width=\textwidth]{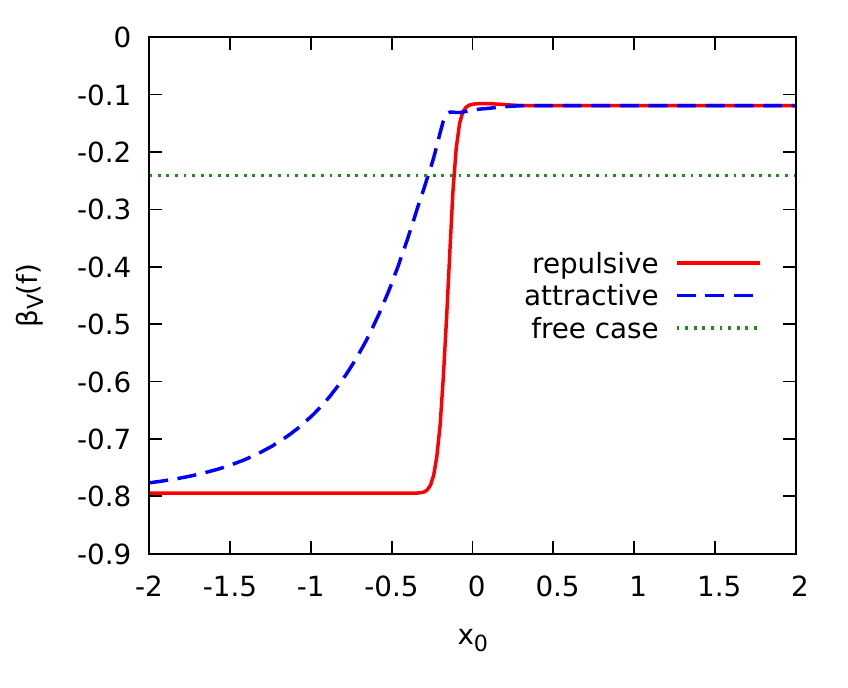}
             \caption{Delta potential, repulsive ($\lambda = 1$) and attractive ($\lambda = -1$)}
             \label{fig:position-delta}
     \end{subfigure}%
     \quad
     \begin{subfigure}[t]{0.45\textwidth}
             \includegraphics[width=\textwidth]{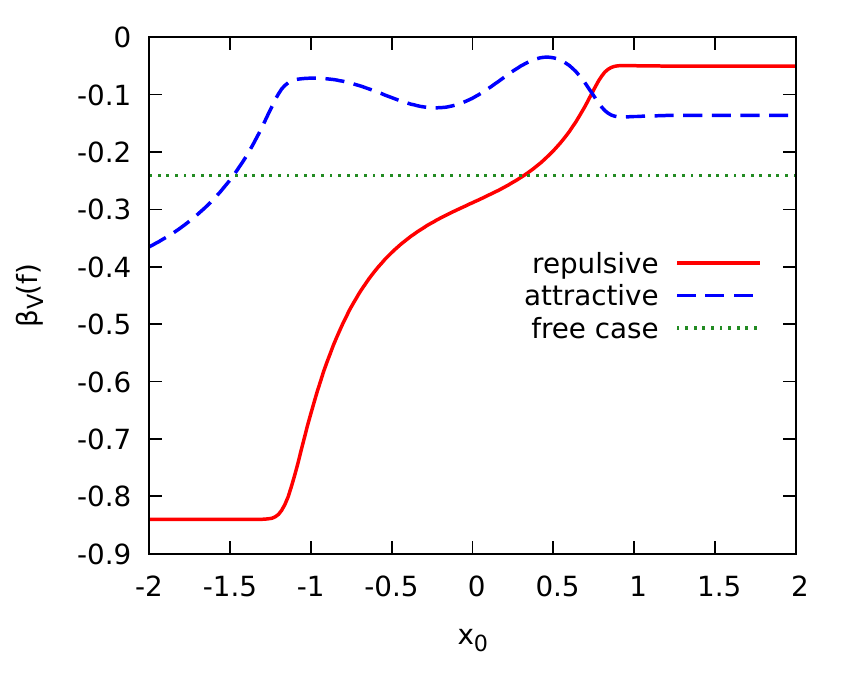}
             \caption{Rectangular potential, repulsive ($\lambda = 2$) and attractive ($\lambda = -2$)}
             \label{fig:position-squarewell}
     \end{subfigure}
     \quad
     \begin{subfigure}[t]{0.45\textwidth}
             \includegraphics[width=\textwidth]{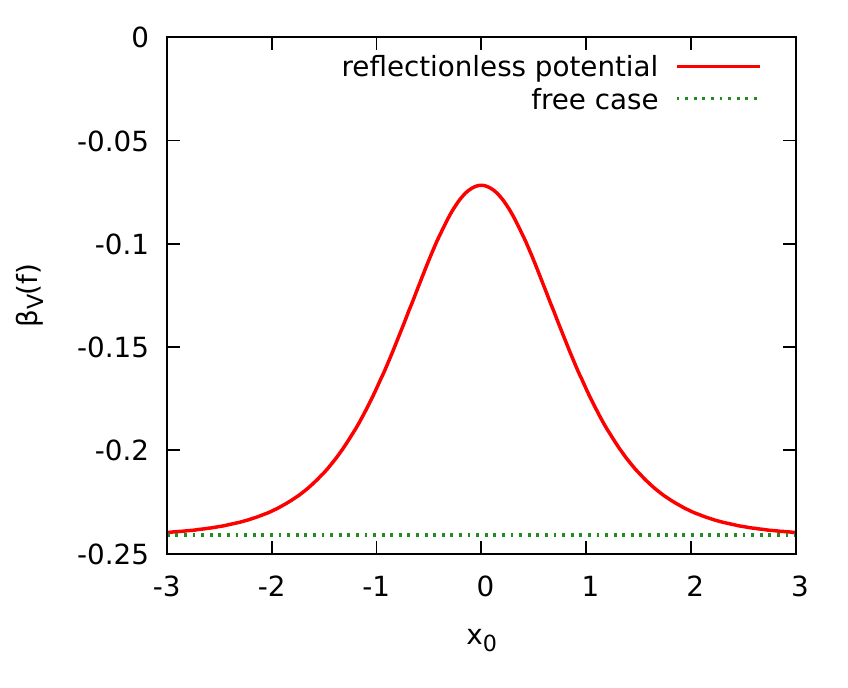}
             \caption{Reflectionless P\"oschl-Teller potential, $\mu=1$}
             \label{fig:position-transparent}
     \end{subfigure}%
     \caption{Backflow bound depending on position. Parameters: $n=1000$, $p_\mathrm{max}=150$, $\sigma = 0.1$.}
     \label{fig:position}
\end{figure}

\subsection{Position of measurement}

Next, we investigate the dependence of the backflow effect on the position of measurement within a potential. That is, we vary the center point $x_0$ of our Gaussian $f$, while its width $\sigma$ remains fixed. 

Fig.~\ref{fig:position-delta} shows the results for a delta potential, both in the attractive and the repulsive case. While these two cases differ, they both have in common that the backflow to the \emph{left} of the potential barrier is larger than in the free case, which can be interpreted as an effect of reflection at the barrier. The backflow to the \emph{right} of the barrier is lower than for a free particle, owing to low-energy contributions being reflected and hence not contributing to the interference effect. For a comparison of the analytic estimate \eqref{eq:single-delta-with-gauss} with the numerical result, we note that for the chosen parameters, \eqref{eq:single-delta-with-gauss} gives  $\beta_V(f) \geq -194.050$ in all cases, which is compatible with Fig.~\ref{fig:position-delta} but certainly a very rough estimate. 

In a rectangular potential, Fig.~\ref{fig:position-squarewell}, we can observe similar effects: higher backflow to the left, lower backflow to the right of the potential. The differences between the repulsive and attractive case are more pronounced however, in particular one observes resonance effects in the interaction region of the attractive potential. Note that the attractive potential in question, with $\lambda=-2$, has two bound states.  

The situation is different in a reflectionless P\"oschl-Teller potential (Fig.~\ref{fig:position-transparent}; we consider $\mu=1$). Here the backflow constant approaches the free value $\approx 0.241$ both left \emph{and} right of the potential, which may be explained by the absence of reflection: the particle behaves like a free one far away from the potential, except for a momentum-dependent phase shift, which does not influence the lowest eigenvalue. Inside the interaction region, the backflow effect is smaller than in the free case, which is not surprising given that the potential is attractive, i.e., classically the particle velocity is higher than in the free case.
The analytic estimate \eqref{poeschl1} yields $\beta_V(f) \geq -283.261$.

\begin{figure}
     \centering
     \begin{subfigure}[t]{0.45\textwidth}
             \includegraphics[width=\textwidth]{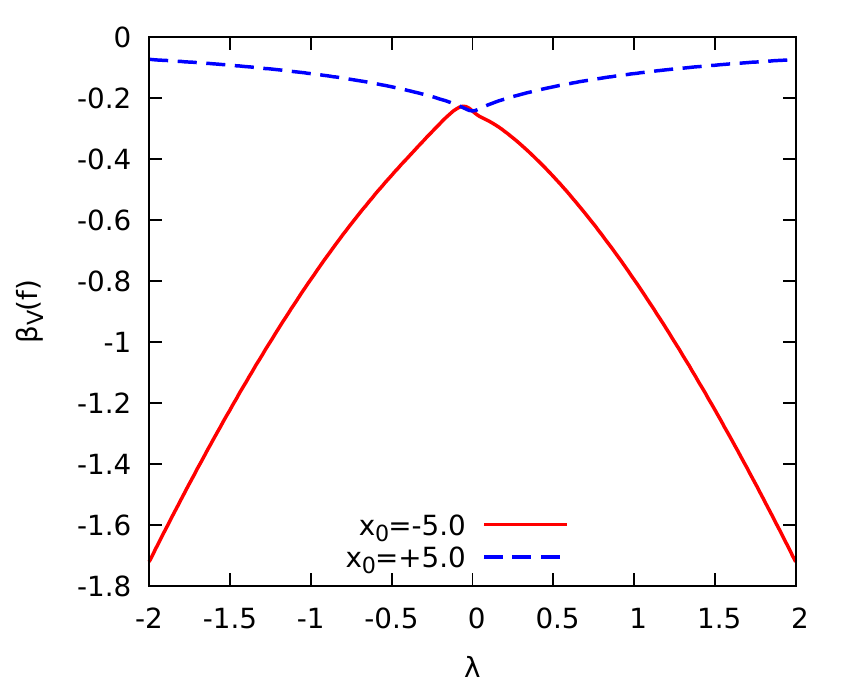}
             \caption{Delta potential}
             \label{fig:amplitude-delta}
     \end{subfigure}%
     \quad
     \begin{subfigure}[t]{0.45\textwidth}
             \includegraphics[width=\textwidth]{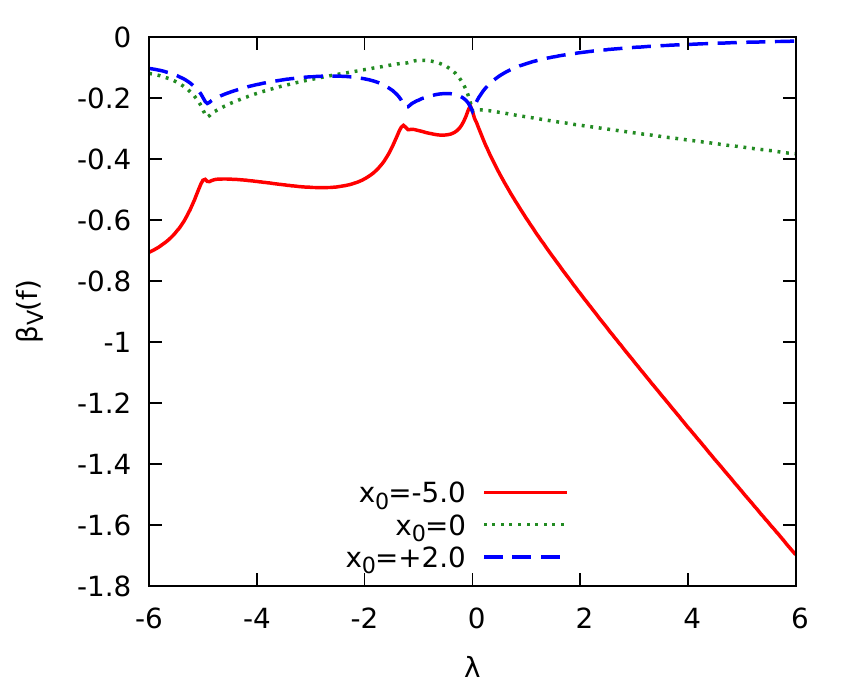}
             \caption{Rectangular potential}
             \label{fig:amplitude-squarewell}
     \end{subfigure}%
     \quad
     \begin{subfigure}[t]{0.45\textwidth}
             \includegraphics[width=\textwidth]{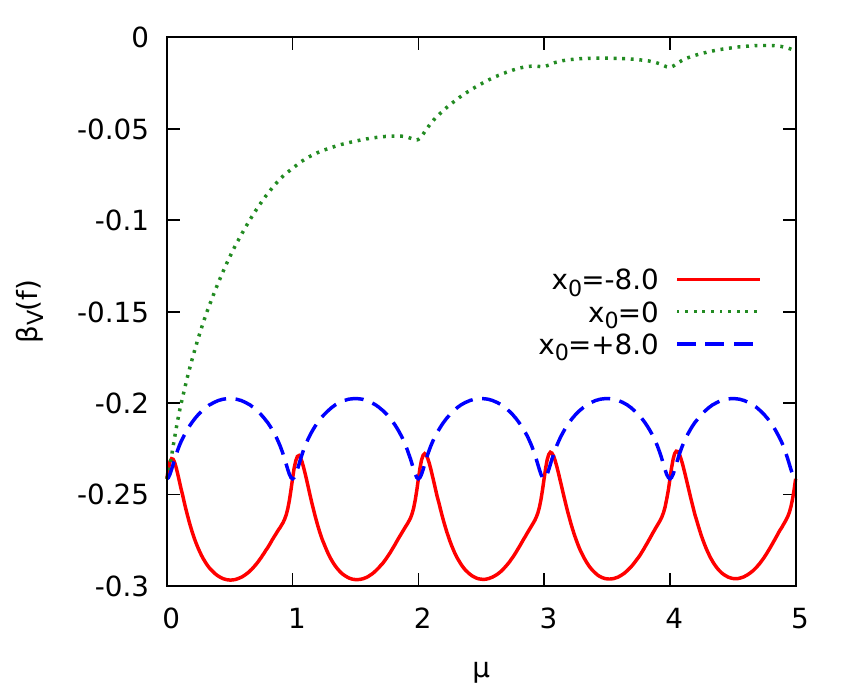}
             \caption{P\"oschl-Teller potential}
             \label{fig:amplitude-poeschlteller}
     \end{subfigure}%
     \caption{Backflow depending on strength of the potential. Parameters: $n=1000$, $p_\mathrm{max}=150$, $\sigma = 0.1$.}
     \label{fig:amplitude}
\end{figure}

\subsection{Strength of potential}

Fig.~\ref{fig:amplitude-delta} shows the backflow far to the left and far to the right of a delta potential with varying amplitude $\lambda$. As expected, backflow on the left of the potential grows with increasing $|\lambda|$, regardless whether attractive or repulsive, while backflow on the right decreases in these situations. The slight asymmetry of the curve near $\lambda=0$ can likely be attributed to the fact that $x_0=-5$ is not sufficiently ``far away'' from the interaction zone in this parameter region.
The analytic estimate \eqref{eq:single-delta-with-gauss}, for fixed $\sigma = 0.1$ and varying $\lambda$, yields $\beta_V(f) \geq -65.347 -128.704\lvert \lambda \rvert$. While this is again very rough in absolute terms, the linear increase for large $\lambda$ appears to match the numeric results.

Next let us turn to the rectangular potential for varying strength $\lambda$, see Fig.~\ref{fig:amplitude-squarewell}. In the repulsive case ($\lambda>0$), the behavior far away from the potential is similar to the delta potential case; the backflow in the interaction region at $x_0=0$ interpolates between the left and right asymptotics. For attractive potentials, however, resonance effects appear to contribute significantly. Note in particular the cusps in the graph near $\lambda \approx -1.2$ and $\lambda \approx -4.9$, which are the points where the number of bound states changes, and hence zero-energy resonances occur. 

Finally,  Fig.~\ref{fig:amplitude-poeschlteller} shows the backflow to the left, to the right, and within the interaction zone of a P\"oschl-Teller potential, for not necessarily integer values of $\mu$. One readily observes that the integer values, where the potential becomes reflectionless, are special in that the backflow far away from the potential matches the free value; for non-integer $\mu$, backflow is generally larger to the left and smaller to the right, which as before can be interpreted as an effect of reflection. The backflow within the interaction zone ($x_0=0$) behaves very differently, and mostly becomes smaller as the strength of the potential increases.

\begin{figure}
     \centering
     \begin{subfigure}[t]{0.42\textwidth}
             \includegraphics[width=\textwidth]{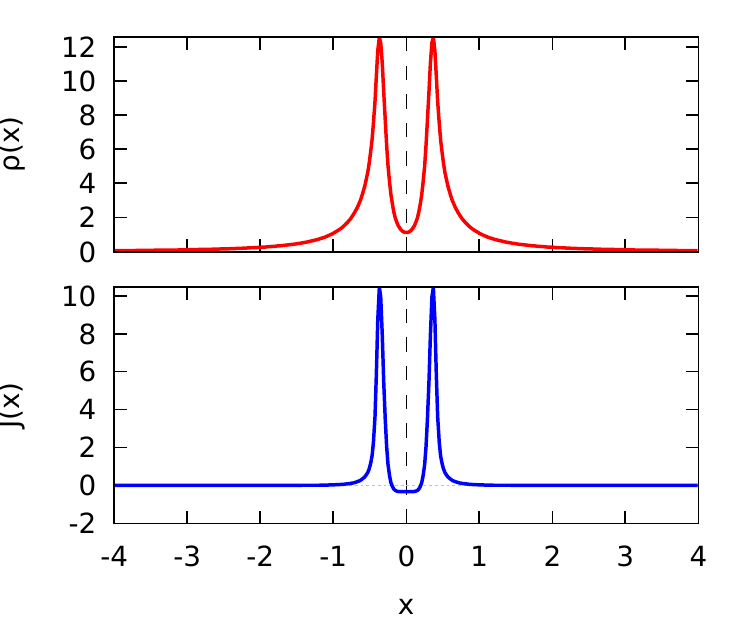}
             \caption{Zero potential, $x_0=0$}
             \label{fig:timeevol-free}
     \end{subfigure} 
     \\
     \begin{subfigure}[t]{0.42\textwidth}
             \includegraphics[width=\textwidth]{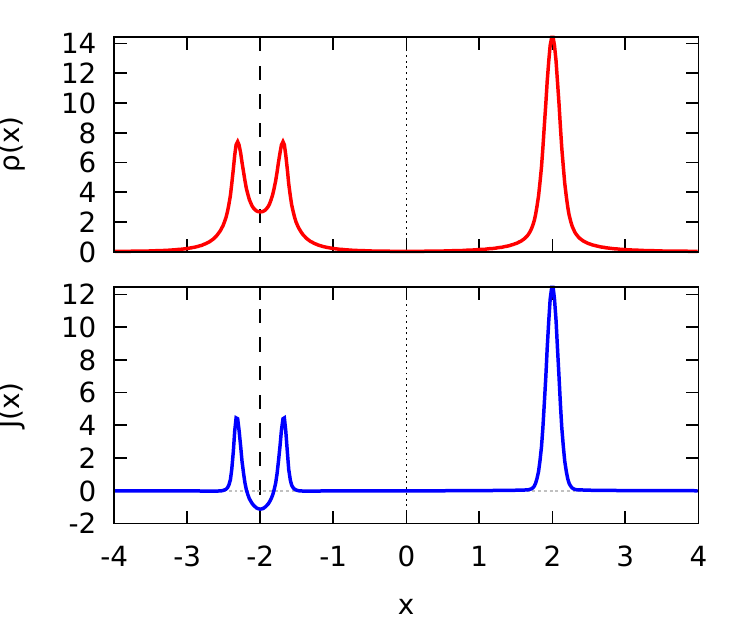}
             \caption{Delta potential, $\lambda=1$, $x_0=-2$}
             \label{fig:timeevol-delta-left}
     \end{subfigure} 
     \\
     \begin{subfigure}[t]{0.42\textwidth}
             \includegraphics[width=\textwidth]{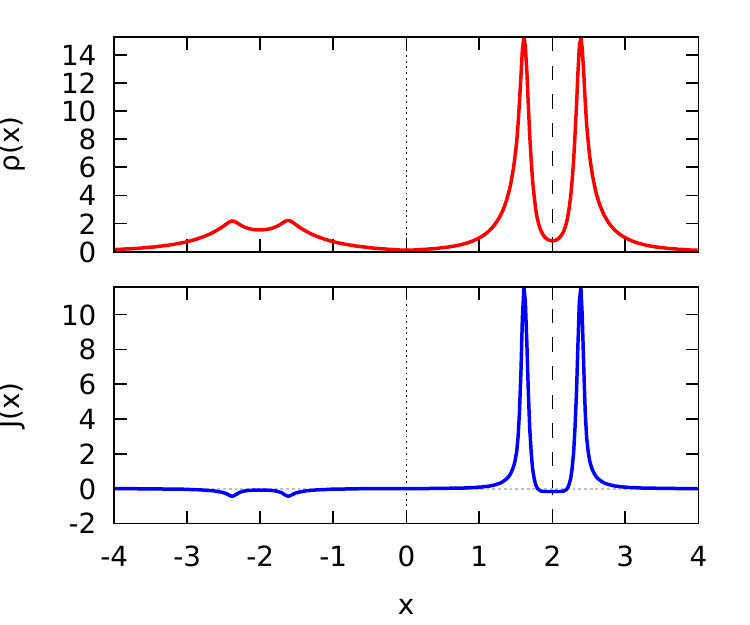}
             \caption{Delta potential, $\lambda=1$, $x_0=+2$}
             \label{fig:timeevol-delta-right}
     \end{subfigure}%
     \caption{Probability distribution and probability current for the maximum backflow vector in configuration space at $t=0$. Dashed vertical lines indicate the position of measurement $x_0$, dotted vertical lines the position of the potential. Parameters: $n=2000$, $p_\mathrm{max}=200$, $\sigma=0.1$. See the animations \cite{anc-arxiv} for evolution with time $t$.}
     \label{fig:timeevol}
\end{figure}

\subsection{Probability current and time evolution}

Lastly, let us look at the shape of the maximum backflow eigenvector $\psi$ in position space. We consider the probability density $\rho(x) = |\psi(x)|^2$ and the probability current $j_\psi(x)=\hscalar{\psi}{ J(x)\psi}$ as a function of $x$, both plotted at time $t=0$ (Fig.~\ref{fig:timeevol}) and as an animation showing the time evolution \cite{anc-arxiv}.

In the free case, Fig.~\ref{fig:timeevol-free}, the situation looks much like the time-smeared backflow eigenvector \cite[Fig.~1]{EvesonFewsterVerch:2003}: The wave packet consists of two forward-moving parts, but there is a negative probability current from the right part to the left part.

For the interacting case, we restrict ourselves to a simple example: a repulsive delta potential, once optimizing for maximum backflow to the left of the potential ($x_0=-2$, Fig.~\ref{fig:timeevol-delta-left})
and once to the right ($x_0=2$, Fig.~\ref{fig:timeevol-delta-right}). It turns out that the behavior is rather similar to the free case, only that the reflected, respectively, transmitted part of the wave function now splits into two wave packets, between which a negative probability current exists.
Note that in Fig.~\ref{fig:timeevol-delta-right}, even if we optimized the vector for maximum backflow around $x_0=+2$, it still happens to exhibit substantial backflow around $x=-2$.

\section{Summary and Outlook}\label{sec:conclusion}

The purpose of this paper was to explore aspects of the backflow effect that go beyond the well-known interaction-free or purely kinematical situation. We have formulated backflow in a general scattering setup by considering states with incoming right-moving asymptotes, interacting with an arbitrary short-range potential. Our results show that the features of the current operator that are typical for backflow in the interaction-free case also persist in the presence of a potential: First of all, the averaged current may produce negative expectation values in asymptotically right-moving states. Moreover, the averaged backflow remains unbounded above but bounded below in this setting, which shows that also the spatial extent of this phenomenon has the same behavior as in the free case.

These findings may be summarized by saying that the main features of backflow are stable under the addition of a potential term to the kinetic Hamiltonian. This stability even holds for arbitrarily strong potentials, meaning that backflow is a universal quantum effect. 

Nevertheless, we saw in examples that the effect becomes more intricate in the presence of a potential. In particular, the maximal amount of backflow $\beta_V(f)$ in scattering situations depends now on the position of the potential $V$ relative to the position of measurement, corresponding to the center of the averaging function $f$. The plots in Fig.~\ref{fig:position} show that far to the right or left of the potential, the maximal backflow converges to a fixed value. Whereas the limit to the right is easy to describe analytically as well, the limit to the left is more complicated: Here an incoming and a reflected wave are superposed, which leads a sum of integral operators, the spectrum of which is difficult to estimate.

The present work was focused entirely on {\em spatial} averages of probability currents. Just as well one could study how {\em temporal} averages respond to the addition of a potential term in the Hamiltonian, thus investigating the Bracken-Melloy constant $\lambda^{H_0}$ for more general time evolutions. Although we did not discuss this point here, let us mention that with the numerical methods at hand, it is easily possible to obtain approximations to $\lambda^{{H_0}+V(X)}$, defined with respect to asymptotic right-movers, once a potential $V$ has been fixed. Non-trivial analytical bounds on this number are however not even known in the free case, and would require new ideas.

We expect that backflow, ultimately being connected to the uncertainty principle, exists in one form or another in various other systems of quantum physics. For example, one could consider a particle with internal degrees of freedom, scattering processes in higher dimensions, or multi-particle systems. On the quantum field theoretic level, so-called quantum energy inequalities describe phenomena which are similar to backflow \cite{EvesonFewsterVerch:2003,Fewster:2012}. To formulate and investigate the whole spectrum of such quantum phenomena in a common framework would, however, require a better understanding of the mathematical core of these effects. We hope to return to this question in the future.

\appendix*
\section{Numerical methods} \label{app:numerics}

The custom computer code which was used to produce the approximations in Sec.~\ref{sec:examples} is supplied with this article \cite{anc-arxiv}, along with documentation. We invite the reader to run it with changed parameters, or indeed to modify the code to accommodate other choices of potentials, etc. Here we briefly describe the numerical methods employed and their relation to the code.

The essential purpose of the code is to approximate the lowest eigenvalues of integral operators, in particular those with the kernel \eqref{eq:kvf}. Let us consider a generic operator $T$ on $L^2(\mathbb{R}_+,dp)$ with smooth kernel $K$ first. Similar to \cite[Sec.~7]{BostelmannCadamuro:oneparticle}, we choose a momentum cutoff $p_\mathrm{max}$, divide the momentum interval $[0,p_\mathrm{max}]$ equally into $n$ subintervals, and choose orthonormal step functions $\tilde\psi_j$ ($j = 0,\ldots, n-1)$ supported on one of these intervals. The operator $T$ is then approximated by 
the matrix $M$ with entries
\begin{equation}\label{eq:mtxe-numeric}
\begin{aligned}
     M_{jk} &= \hscalar{\psi_j}{T \psi_k} = \int dq\,dq\, \tilde\psi_j(p) \,K(p,q) \,\tilde\psi_k(q) 
     \\ &\approx \frac{p_\mathrm{max}}{n} K(p_j, p_k)\,,
\end{aligned}
\end{equation}
where $p_j = (j+\frac{1}{2}) p_\mathrm{max}/n$.
We find the lowest eigenvalue and -vector of $M$, and hence $T$, using the inverse power method: Given an initial lower bound $\lambda_0$ for the operator, and a generic guess $\xi_0$ for the lowest eigenvector, we compute the sequence $(M-\lambda_0)^{-m} \xi_0$, which for $m \to \infty$ converges to the desired lowest eigenvector after normalization. In fact, in order to obtain a good initial estimate, we use this iteration twice, once with a rough guess for $\lambda_0$ and with moderate $n$, and then for a larger $n$, with $\lambda_0$ estimated from the first run. (See \code{kernels.SpectrumTools}.)

Thus we have reduced the question to evaluating the kernel $K$, which in the case of the probability current is given by $K_{V,f}$ in \eqref{eq:kvf}. To evaluate $K_{V,f}$, we need specific information about the potential, namely, the function $\varphi_k$ and its derivative $\partial_x \varphi_k$. (This is modeled by the abstract class \code{models.ScatteringModel} in the code.) Given this, we can evaluate the integral in \eqref{eq:kvf} using Simpson's rule. 

However, the most efficient way of evaluating $\varphi_k$ and its derivative is highly dependent on the potential in question. For the delta as well as rectangular potentials, the explicit solutions of the Schr\"odinger equation are well known and can be used directly, although their discontinuities need to be taken into account in the numerical integration (see \code{models.DeltaPotentialModel} and \code{models.RectangularPotentialModel}). The same holds for the P\"oschl-Teller potential with integer $\mu$; we use this fact for $\mu=1$  (\code{models.SimpleTransparentModel}). For a generic potential, and in particular for P\"oschl-Teller with fractional $\mu$, we solve the Schr\"odinger equation \eqref{schroe} numerically. To that end, we consider the equivalent equation
\begin{equation}
  \partial_x^2 \chi (k,x) = 2 V(x) \chi(k,x) - 2 i k \partial_x \chi(k,x)   
\end{equation}
for the function $\chi(k,x) := \varphi_k(x)e^{-ikx}/T(k)$, and rewrite it as a system of four real-valued first order equations.  
We then solve this ordinary differential equation (ODE) system numerically at fixed $k$ using  an adaptive Runge-Kutta scheme, specifically, the Dormand-Prince method of order 8(5,3) in the form of \cite{Hairer:ode1} as implemented in \cite{commonsmath361}. The initial conditions are $\chi(k,x)=1$, $\partial_x\chi(k,x)=0$ far to the right of the potential. The result for $\chi$ can be cached for each (discretized) $k$, limiting the impact of the numerical ODE solver on overall computation time. See the class \code{models.GenericPotentialModel} for details.

In some cases, numerical integration in \eqref{eq:kvf} can be avoided if $\varphi_k$ is a linear combination of plane waves, and the Fourier transform of $f$ is explicitly known, as in the case of a Gaussian. We make use of this for compactly supported potentials (delta and rectangular) when the position of measurement $x_0$ is far to the left or to the right of the potential; see \code{kernels.AsymptoticCurrentKernel}.

\footnotesize
\bibliography{backflow}
\bibliographystyle{unsrt}

\end{document}